\documentclass[journal]{IEEEtran}

\ifCLASSINFOpdf
\else
   \usepackage[dvips]{graphicx}
\fi
\usepackage{url}

\usepackage[utf8]{inputenc} 
\usepackage{graphicx}
\usepackage[ruled,vlined]{algorithm2e}
\usepackage{amsmath}

\newcommand\numberthis{\addtocounter{equation}{1}\tag{\theequation}}

\usepackage{amsfonts}
\usepackage{amssymb}
\usepackage{amsthm}
\usepackage{dsfont}
\usepackage{array}
\usepackage{multirow}
\usepackage{lipsum}
\usepackage{mathtools}
\usepackage{cuted}

\usepackage[caption=false]{subfig}
\usepackage[noadjust]{cite}
\usepackage{bm}
\newcommand{\bs}{\ensuremath{\bm{s}}}
\newcommand{\bpi}{\ensuremath{\bm{\pi}}}
\newcommand{\bt}{\ensuremath{\bm{t}}}
\newcommand{\by}{\ensuremath{\bm{y}}}
\newcommand{\bn}{\ensuremath{\bm{n}}}
\newcommand{\R}{\mathbb{R}}

\newcommand{\pS}{\phantom{*}}

\newtheorem{theorem}{Theorem}

\usepackage{xcolor}

\usepackage{stfloats}
\usepackage{rotating} 

\usepackage{lineno}
\usepackage{siunitx}
\usepackage{etoolbox}
\newcommand{\ubold}{\fontseries{b}\selectfont}
\robustify\ubold

\newcommand{\pend}{PENDANTSS\xspace}

\usepackage[final]{changes}

 \setdeletedmarkup{} 

\begin{document}
\bstctlcite{MyBSTcontrol}
\title{
\pend: PEnalized Norm-ratios Disentangling  Additive Noise, Trend and Sparse Spikes
}

\author{Paul Zheng, {\em Student Member, IEEE},  Emilie Chouzenoux, {\em Senior Member, IEEE}, Laurent Duval, {\em Senior Member, IEEE}

\thanks{This work was supported by the European
Research Council Starting Grant MAJORIS ERC-2019-STG-850925.}
\thanks{P. Zheng is currently with  Chair   of   Information   Theory   and   Data   Analytics,   RWTH   Aachen   University, Germany (paul.zheng@inda.rwth-aachen.de); work  conducted while   at Univ. Paris-Saclay, CentraleSup\'elec, CVN, Inria, Gif-sur-Yvette, France.}
\thanks{E. Chouzenoux is with Univ. Paris-Saclay, CentraleSup\'elec, CVN, Inria, Gif-sur-Yvette, France (emilie.chouzenoux@centralesupelec.fr).}
\thanks{L. Duval is with IFP Energies nouvelles, France (laurent.duval@ifpen.fr).}
}

\maketitle

\begin{abstract} 
Denoising, detrending, deconvolution: usual  restoration tasks, traditionally decoupled.  Coupled formulations entail  complex ill-posed inverse problems. We propose \pend for joint trend removal and blind deconvolution of sparse  peak-like signals. It blends a parsimonious  prior with the hypothesis that smooth trend and  noise  can somewhat be separated by low-pass filtering. We  combine the generalized quasi-norm  ratio SOOT/SPOQ\footnote{SOOT: Smoothed One-Over-Two / SPOQ: Smoothed $p$-Over-$q$.}
 sparse penalties $\ell_p/\ell_q$    with  the BEADS\footnote{BEADS: Baseline Estimation And Denoising with Sparsity.} ternary-assisted  source separation algorithm. This results in a both   convergent and efficient tool, with  a novel Trust-Region block alternating variable metric forward-backward approach. It outperforms comparable methods, when applied to typically peaked analytical chemistry signals. Reproducible code is provided.
\end{abstract}


\IEEEpeerreviewmaketitle

\section{Introduction and background}
Restoration  recovers  information from observations with amplitude distortion, level displacement or  random disturbance. 
We seek estimates $\widehat{\bs}$, $\widehat{\bt}$ and 	$\widehat{\bpi}$  from observation $\bm{y}$, under the discrete additive-convolutive degradation:
\vspace{-.25cm}
\begin{equation}
	\bm{y} 
	=   \overline{\bs}\ast \overline{\bpi} + \overline{\bt} + \bn \,.
	\label{eq:observation_signal}
	\vspace{-.13cm}
\end{equation}
Among  $N$ sample values, a series of  \emph{spikes}   (also called  impulses, events, ``diracs" or spectral lines)  models  the {\bf{first component}}, the sought sparse signal $ \overline{\bs}\in \mathbb{R}^N$. 
Its  convolution  with an unknown short-support \emph{kernel} $\overline{\bpi}\in \mathbb{R}^L$ --- typically peak-shaped ---    yields the  \emph{peak-signal} $\overline{\bm{x}}=\overline{\bs}  \ast\overline{\bpi} \in \mathbb{R}^N$.
{The \bf{second  component}} $\overline{\bt} \in \mathbb{R}^N$ offsets the reference level, harming quantitative estimations. It can be called baseline, background, continuum, drift,  or wander. 
We opt for \emph{trend}, a reference above which peaks are detected, evaluated and measured. ``Trends" address   slowly varying amplitude shifts (due to seasonality, calibration distortion, sensor  decline\ldots), challenging its automated removal.  
{\bf{Third  component}}  $\bn \in \mathbb{R}^N$ ({\emph{noise}}) gathers  stochastic residuals. Given~\eqref{eq:observation_signal}, the goal is to perform jointly denoising, detrending and deconvolution.  Namely,  given $\bm{y}$, retrieve estimations of the spiky signal, the kernel and the trend. Fig.~\ref{Fig_Dataset} is reminiscent of standard spectral subtraction \cite{Boll_S_1979_j-ieee-tassp_sup_ansss}, and  motivated here by  peak-signal retrieval  in separative analytical chemistry (AC): chromatography, spectrometry, spectroscopy \cite{Lynch_J_2003_book_physico-chemical_aicpgc}, where   peak localization, amplitude, width or area  provide useful chemical quantitative information. 

Whether acquired in its natural domain \cite{Marmin_A_2021_j-sp_sparse_srnmpro} or after sparsification \cite{Gauthier_J_2009_j-ieee-tsp_optimization_socfb},  noise/trend/spike models \eqref{eq:observation_signal} cover many  multidimensional  issues: signal (1D), image (2D),  video, volume (3D+). 
We focus here  on 1D  data common to  diverse domains:  Fourier spectral analysis, econometrics, stock prices, biomedical measurements (ECG, EEG, EMG), environmental observations, astronomical spectroscopy, etc.

    On the one hand, joint denoising and detrending is a long-standing preprocessing question, ranging from time series analysis to imaging. Background issues are commonly solved using a host of filling, fitting and filtering methods. We refer to overviews in \cite{Duval_L_2013_p-icassp_overview_spics,Zhao_Z_2021_j-sp_robust_etfun}, and  for  AC to  background corrections backcor \cite{Mazet_V_2005_j-chemometr-intell-lab-syst_background_rsdmnqcf} and BEADS \cite{Ning_X_2014_j-chemometr-intell-lab-syst_chromatogram_bedusbeads}.

On the other hand, joint denoising and blind deconvolution matters from channel estimation in communications \cite{Amari_S_1997_p-ieee-spawc_multichannel_bdeung}
to image deblurring \cite{Krishnan_D_2011_p-cvpr_blind_dnsm}. We refer to \cite{Chaudhuri_S_2014_incoll_blind_dmr,Sun_Q_2021_PREPRINT_convex_sbd}, and especially emphasize on sparsity-promoting methods like 
SOOT \cite{Repetti_A_2015_j-ieee-spl_euclid_tsbdsl1l2r} and SPOQ \cite{Cherni_A_2020_j-ieee-tsp_spoq_lpolqrssrams},  using smoothed ``scale-invariant" norm ratios.

 \pend original 
 contributions are (i) a fully coupled and solvable non-convex formulation for \eqref{eq:observation_signal} (Section  \ref{Sec_Problem}) and (ii) a novel efficient joint disentangling  algorithm (forward-backward-based \cite{Bolte_J_2014_j-math-programm_proximal_almnnp, Chouzenoux_E_2016_j-global-optim_block_cvmfba}) with proved convergence (Section \ref{Sec_Algorithm}), 
 validated by its comparative performance  (Section  \ref{Sec_Results}).

\vspace{-.2cm}
\section{Proposed problem formulation\label{Sec_Problem}}

\subsection{BEADS peak/trend/noise separation paradigm}

Estimates of $(\widehat{\bs},\widehat{\bt},	\widehat{\bpi})$ of $(\overline{\bs},\overline{\bt},\overline{	\bpi})$ are obtained through the resolution of the penalized least squares problem
\vspace{-.1cm}
\begin{equation}
   \text{minimize}_{\bs,\bt\in \R^N  \atop 
	\bpi \in \R^L}\ \frac{1}{2}\|\by - \bpi * \bs -\bt\|^2 + R(\bs,\bt,\bpi),
	\label{pb:original}
	\vspace{-.1cm}
\end{equation}
with regularization term $R$ incorporating prior knowledge.
 Disentangling trend and signal is tedious \cite{Selesnick_I_2014_j-ieee-tsp_simultaneous_lpftvd}.  
As in BEADS \cite{Ning_X_2014_j-chemometr-intell-lab-syst_chromatogram_bedusbeads}, we assume that the trend can be recovered from a peakless observation through a low-pass filter~$\bm{L}$:
	\vspace{-.15cm}
\begin{equation}
	\widehat{\bt} = \bm{L}(\by - \widehat{\bpi}*\widehat{\bs}).
	\label{eq:t_hat_orig}
		\vspace{-.15cm}
\end{equation}
This motivates the rewriting of the data fidelity term in \eqref{pb:original} as:
\vspace{-.2cm}
\begin{align}
(\forall \bs \in \mathbb{R}^N) (\forall \bpi \in \mathbb{R}^L)  \;
\rho(\bs,\bpi) & 
= \frac{1}{2}\|\by - \bm{Ly} - \bm{H} (\bpi * \bs)\|^2 \nonumber \\
& = \frac{1}{2}\|\bm{H}(\by-\bpi * \bs)\|^2,
\label{eq:rho}
\end{align}
where~$\bm{H}=\mathrm{\mathbf{Id}}_N-\bm{L}$ is a high-pass filter, and $\mathrm{\mathbf{Id}}_N$ the identity operator of $\mathbb{R}^N$. We introduce a regularization term~$\Psi$, promoting signal sparsity. We add two extra terms to constrain estimates $\widehat{\bs}$ and $\widehat{\bpi}$ to sets $C_1 \subset \mathbb{R}^N$ and $C_2 \subset \mathbb{R}^L$ assumed  closed, non-empty and convex. 
{The indicator function $\iota_{C_i}$, $i \in \{1,2\}$ equals zero when the value evaluated belongs to~$C_i$, $+\infty$ otherwise. } {O}ptimization problem~\eqref{pb:original}  becomes:
\begin{equation}
	\underset{\bs\in\R^N,\,\bpi\in\R^L}{\text{minimize}}\ \frac{1}{2}||\bm{H}(\by - \bpi * \bs)||^2 +\iota_{C_1}(\bs) + \iota_{C_2}(\bpi)+ \lambda \Psi(\bs).
	\label{pb: pb2solve}
\end{equation}
The estimated  {trend} can be obtained  {from \eqref{eq:t_hat_orig} with $\widehat{\bpi}$ and $\widehat{\bs}$ obtained by \eqref{pb: pb2solve}}. 
\vspace{-.1cm}
\subsection{{SPOQ/SOOT norm/quasi-norm ratio penalties}}
Tractable penalties for sparsity characterization include homogeneous $\ell_p$-norms, quasi-norms (for $0<p<1$), or mixed norms. We refer to \cite{Fu_H_2006_j-siam-j-sci-comput_efficient_mmml2l1l1l1nir,Repetti_A_2015_j-ieee-spl_euclid_tsbdsl1l2r,Soubies_E_2017_j-siam-j-optim_unified_vecpl2l0m,Cherni_A_2020_j-ieee-tsp_spoq_lpolqrssrams,Sun_Q_2021_PREPRINT_convex_sbd} and references therein.
Ratios of norms are also promising proxies, being scale-invariant \cite{Hurley_N_2009_j-ieee-tit_comparing_ms}. \replaced{We here promote sparse $\widehat{\bs}$ through the family of SPOQ norm ratio penalties}{We promote sparse solutions $\bs$  by the 
  Smoothed~$p$-Over-$q$ (SPOQ) family of penalties}, introduced in \cite{Cherni_A_2020_j-ieee-tsp_spoq_lpolqrssrams}, as a generalization to the  SOOT ratio  \cite{Repetti_A_2015_j-ieee-spl_euclid_tsbdsl1l2r}.  Let~$p\in ]0,2[$ and~$q\in [2,+\infty[$. Smoothed approximations to the $\ell_p$ quasi-norm and~$\ell_q$ norm, parameterized by constants~$(\alpha, \eta)\in]0,+\infty[^2$ are defined, for every  $\bm{s} = (s_n)_{1\leq n\leq N}\in \R^N$, as:
\vspace{-.1cm}
\begin{equation}
	\ell_{p,\alpha}(\bm{s}) = \left( \sum_{n=1}^{N} \left((s_n^2 + \alpha^2)^{p/2}-\alpha^p\right) \right)^{1/p},
	\vspace{-.1cm}
\end{equation}
and
\vspace{-.1cm}
\begin{equation}
	\ell_{q,\eta}(\bm{s}) = \left( \eta^q + \sum_{n=1}^{N} |s_n|^q\right)^{1/q}.
	\vspace{-.1cm}
\end{equation}
 {The non-convex} SPOQ penalty is  {given}, for  $\beta\in]0,+\infty[$, as:
\begin{equation}
(\forall \bs \in \mathbb{R}^N) \quad
\Psi(\bm{s}) = \log\left(\frac{(\ell_{p,\alpha}^p(\bm{s}) + \beta^p)^{1/p}}{\ell_{q,\eta}(\bm{s})} \right).
\vspace{-.1cm}
\end{equation}
$\Psi$ is Lipschitz differentiable on $\mathbb{R}^N$ \cite[Prop.~2]{Cherni_A_2020_j-ieee-tsp_spoq_lpolqrssrams} { and admits $\bm{0}_N$ as a local minimizer}  {when} \cite[Prop.~1]{Cherni_A_2020_j-ieee-tsp_spoq_lpolqrssrams}:
\begin{equation}
\label{cond:SPOQ}
    q>2, \quad \text{or} \quad q = 2 \quad \text{and} \quad \eta^2\alpha^{p-2} > \beta^p.
\end{equation}
Condition \eqref{cond:SPOQ} is assumed throughout this paper. 

\section{Proposed optimization algorithm\label{Sec_Algorithm}}

\subsection{Problem structure}

The objective function in~\eqref{pb: pb2solve} is the sum of a differentiable function (least squares + SPOQ) and  {terms acting separably on} $\bs$ or $\bpi$ (i.e., indicator terms).  {In the differentiable part}
\begin{equation}
(\forall \bs \in \mathbb{R}^N) (\forall \bpi \in \mathbb{R}^L)  \quad
	f(\bs,\bpi ) = \rho(\bs,\bpi) + \lambda \Psi(\bs) {,}
	\label{eq: bd_f_simp}
\end{equation}
 {with function $\rho$ from \eqref{eq:rho} quadratic in $\bs$ and $\bpi$}. In particular, for every $\bpi \in \mathbb{R}^L$ (resp.~$\forall \bs \in \mathbb{R}^N$), the gradient~$\nabla \rho_1(\cdot,\bpi)$ (resp.~$\nabla \rho_2(\bs,\cdot)$) of $\rho$ with respect to its first  (resp. second) variable is Lipschitz continuous with constant~$\Lambda_1(\bpi)$ (resp.~$\Lambda_2(\bs)$). As aforementioned, $\nabla \Psi$ is Lipschitz continuous too. The second part of the objective function reads as:
\begin{equation}
(\forall \bs \in \mathbb{R}^N) (\forall \bpi \in \mathbb{R}^L)  \quad
	g(\bs,\bpi)= \iota_{C_1}(\bs) +\iota_{C_2}(\bpi).
	\label{eq:g}
\end{equation}
In a nutshell, Problem \eqref{pb: pb2solve} amounts to minimizing: 
\begin{equation}
(\forall \bs \in \mathbb{R}^N) (\forall \bpi \in \mathbb{R}^L)  \quad
	\Omega(\bs,\bpi)= f(\bs,\bpi) + g(\bs,\bpi). \label{eq:obj}
\end{equation}

\subsection{Proposed  {Trust-Region}  {\pend} algorithm}
The  {structure of  \eqref{eq:obj} suggests a block alternating approach where signal $\bs$  and kernel $\bpi$  are updated sequentially. We hereby introduce Algorithm 1, that generalizes the  BC-VMFB algorithm \cite{Chouzenoux_E_2016_j-global-optim_block_cvmfba}, also used in \cite{Repetti_A_2015_j-ieee-spl_euclid_tsbdsl1l2r} for blind deconvolution}.

\begin{algorithm}
	\SetAlgoLined
	\textbf{Settings:} $K_{\max}>0$, $\varepsilon>0$, $\mathcal{I} > 0$,~$\theta\in]0,1[$, $(\gamma_{s,k})_{k \in \mathbb{N}} \in [\underline{\gamma}, 2-\overline{\gamma}]$ and  $(\gamma_{\pi,k})_{k \in \mathbb{N}}  \in [\underline{\gamma}, 2- \overline{\gamma}]$ for some  $(\underline{\gamma}, \overline{\gamma}) \in ]0,+\infty[^2$, $(p,q) \in ]0,2[\times [2,+\infty[$ satisfying \eqref{cond:SPOQ}, convex sets $(C_1,C_2) \subset \mathbb{R}^N \times \mathbb{R}^L$.\\
	\textbf{Initialize:}~$\bm{s}_0\in C_1$,~$\bm{\pi}_0\in C_2$\\
	
	\For{$k = 0, 1,\ldots$}{
	\underline{\emph{Update of the signal}}\\
		\For{$i = 1,\ldots, \mathcal{I}$}{
			Set TR radius $\rho_{k,i}$ using~\eqref{eq:TRradius} with parameter $\theta$\;
			Construct MM metric $\bm{A}_{1,\rho_{k,i}}(\bm{s}_{k},\bm{\pi}_{k})$ using~\eqref{eq:A_qrho};\\ 
			Find $\bm{s}_{k,i} \in C_1$ such that 
 \eqref{eq: calc_s} holds.\\
			\If{$\bm{s}_{k,i}\in \overline{\mathcal{B}}_{q,\rho_{k,i}}$}{Stop loop}
		}
		$\bm{s}_{k+1} = \bm{s}_{k,i}$\;
		\underline{\emph{Update of the kernel}}\\
		Find $\bm{\pi}_{k+1} \in C_2$ such that \eqref{eq: calc_pi} holds.\\
		\underline{\emph{Stopping criterion}}\\
	
	\If{$\|\bs_k -\bs_{k+1}|| \leq \varepsilon$ or $k \geq K_{\max}$}{Stop loop}
	}
$(\widehat{\bs},\widehat{\bpi}) = (\bs_{k+1},\bpi_{k+1})$\,\text{and}\,
$	\bm{\widehat{t}}$ given by \eqref{eq:t_hat_orig}\;
	\KwResult{$\widehat{\bs},\widehat{\bpi}, \bm{\widehat{t}}$}
	\caption{TR-BC-VMFB for solving \eqref{pb: pb2solve}}
	\label{algo:TR-BC-VMFB}
\end{algorithm}

\subsubsection{ {Signal update}}
 Let $k \in \mathbb{N}$ and $(\bs_k,\bpi_k) \in C_1 \times C_2$. The computation of $\bs_{k+1}$ follows one  {Majoration-Minimization (MM) iteration  \cite{Sun_Y_2017_j-ieee-tsp_majorization-minimization_aspcml}}. First, we build a majorization for $\Omega(\cdot,\bpi_k)$ around $\bs_k$. Second, $\bs_{k+1}$ is defined as a minimizer  {to} the majorant. In practice, both steps can be approximated  {for speedup and robustness to numerical errors}.  {A}s emphasized in~\cite{Chouzenoux_E_2022_j-math-imaging-vis_local_mmsmscvpir,Cherni_A_2020_j-ieee-tsp_spoq_lpolqrssrams}, we  need the majorization to be valid only within a neighborhood of the current iterate. For  $\rho\in[0,+\infty[$, the~$\ell_q$-ball complement set  is:
\begin{equation}
	\overline{\mathcal{B}}_{q,\rho}   =\{\bm{s}=(s_n)_{1\leq n\leq N} \in\R^N | \sum_{n=1}^{N} |s_n|^q \geq \rho^q \}.
	\label{eq:TR_ballex}
\end{equation}
From~\cite[Prop. 2]{Cherni_A_2020_j-ieee-tsp_spoq_lpolqrssrams}, we can show that
\begin{multline}
\label{eq:maj1}
(\forall \bs \in \overline{\mathcal{B}}_{q,\rho} \cap C_1) \quad
    \Omega(\bs,\bpi_k) \leq f(\bs_k,\bpi_k) \\+ {(\bs - \bs_k)}^\top \nabla_1 f(\bs_k,\bpi_k)  + \frac{1}{2} \| \bs - \bs_k\|^2_{\bm{A}_{1,\rho}(\bs_k,\bpi_k)},
\end{multline}
where we define the so-called MM metric as:
\begin{multline}
\bm{A}_{1,\rho}(\bm{s}_k,\bpi_k) = (\Lambda_1(\bpi_k) + \lambda \chi_{q,\rho}) \mathrm{\mathbf{Id}}_N + \\ \frac{\lambda}{\ell_{p,\alpha}^p(\bs_k) + \beta^p}\text{Diag}((s_{n,k}^2+\alpha^2)^{p/2-1})_{1\leq n\leq N} ,
\label{eq:A_qrho}
\end{multline}
with the constant $\chi_{q,\rho} = (q-1)/{(\eta^q+\rho^q)^{2/q}}$.
In \eqref{eq:maj1}, $\|.\|_{\bm{A}}$ denotes the weighted Euclidean norm related to a symmetric definite positive (SDP) matrix~$\bm{A}\in\R^{N\times N}$, i.e., $\forall \bm{z}\in \R^N,\ \|\bm{z}\|_{\bm{A}} = (\bm{z}^\top \bm{A}\bm{z})^{1/2}$. Since inequality \eqref{eq:maj1} only holds on a limited region, we introduce a Trust-Region-based (TR) loop   \cite{Conn_A_2000_book_trust-region_m,Chouzenoux_E_2022_j-math-imaging-vis_local_mmsmscvpir} to make sure that the minimizer of the majorant is indeed in the validity domain of \eqref{eq:maj1}. Namely, we set~$\mathcal{I} >0$, a maximum number of trials of TR approach. For $i \in \{1,\ldots,\mathcal{I}\}$, we define the TR radius as:
\begin{equation}
\vspace{-0.2cm}
	\rho_{k,i} = \begin{cases}
		\sum_{n=1}^{N}|s_{n,k}|^q & \text{if } i = 1\,,\\
		\theta \rho_{k,i-1} & \text{if } 2\leq i\leq \mathcal{I}-1\,,\\
		0 & \text{if } i = \mathcal{I}\,.
	\end{cases}
	\label{eq:TRradius}
\end{equation}
We compute the associated MM metric $\bm{A}_{1,\rho_{k,i}}(\bm{s}_k,\bpi_k)$ and define $\bm{s}_{k,i}$ as a minimizer of the right term in \eqref{eq:maj1}. The loop stops whenever $\bm{s}_{k,i}$ belongs to~$ \bar{\mathcal{B}}_{q,\rho_{k,i}}$, which is ensured to arise in a finite number of steps according to \cite{Cherni_A_2020_j-ieee-tsp_spoq_lpolqrssrams}. There remains to explain how we practically compute $\bm{s}_{k,i}$. Depending on the choice for $C_1$, the right term in \eqref{eq:maj1} might not have a closed-form minimizer. Actually, as we will show, it appears sufficient for convergence purpose to search for $\bm{s}_{k,i} \in C_1$ satisfying the first order optimality conditions:
\begin{equation}
\begin{cases}
  (\bm{s}_{k,i}\!  -\! \bm{s}_k)^\top\!\nabla_1 f(\bm{s}_k, \bm{\pi}_k)\! +\! \gamma_{s,k}^{-1}||\bm{s}_{k,i}\!  -\! \bm{s}_k||^2_{\bm{A}_{1,\rho_{k,i}}(\bs_k,\bpi_k)}\!\leq\! 0,
    \label{eq: calc_s}
\\
    ||\nabla_1 f(\bm{s}_k, \bm{\pi}_k)\!+\!\bm{r}_{k,i}^{(1)} || \leq \kappa_1 || \bm{s}_{k,i}\! -\! \bm{s}_k||_{\bm{A}_{1,\rho_{k,i}}(\bs_k,\bpi_k)} 
    \end{cases}
\end{equation}
for some~$\bm{r}_{k,i}^{(1)}~\in  N_{C_1}(\bm{s}_{k,i})$ (i.e., the normal cone of $C_1$ at $\bm{s}_{k,i}$ \cite{Bauschke_H_2011_book_convex_amoths}), and some $\kappa_1>0$. The existence of such an $\bm{s}_{k,i}$ can be shown from \cite[Rem. 3.3]{Chouzenoux_E_2014_j-optim-theory-appl_variable_mfbamsdfcf}. In particular, a minimizer over $C_1$ of the right term in \eqref{eq:maj1} satisfies~\eqref{eq: calc_s}.

\subsubsection{ {Kernel update}}  {It} follows a similar approach. The main difference is that we do not  use  the TR loop in that case, as the function to minimize here is simpler. Let $k \in \mathbb{N}$, and $(\bs_{k+1},\bpi_k) \in C_1 \times C_2$. \replaced{By descent lemma,}:
\begin{multline}
\label{eq:maj2}
(\forall \bpi \in C_2) \quad
    \Omega(\bs_{k+1},\bpi) \leq f(\bs_{k+1},\bpi_k) \\+ ( \bpi - \bpi_k)^\top \nabla_2 f(\bs_{k+1},\bpi_k)  + \frac{\Lambda_2(\bs_{k+1}) }{2} \| \bpi- \bpi_k\|^2 .
\end{multline}
The new iterate $\bpi_{k+1}$ is then defined as a minimizer of the right term of \eqref{eq:maj2}. Hereagain, we can solve this problem in an inexact manner, that is to search for some $\bpi_{k+1} \in C_2$ satisfying
\begin{equation}
\begin{cases}
 (\bpi_{k+1}  - \bpi_{k})^\top\nabla_2 f(\bm{s}_{k+1}, \bpi_{k})\\ 
 \qquad \quad +\gamma_{\pi,k}^{-1} \Lambda_2(\bs_{k+1})\|\bm{\pi}_{k +1}  - \bm{\pi}_{k }\|^2 \leq 0, \numberthis  \label{eq: calc_pi}
\\
    \|\nabla_2 f(\bm{s}_{k+1}, \bm{\pi}_{k})+\bm{r}_{k}^{(2)} \| \leq \kappa_2 \sqrt{\Lambda_2(\bs_{k+1})} \| \bm{\pi}_{k+1} - \bm{\pi}_{k}\|, \nonumber
\end{cases}
\end{equation}
 for some~$\bm{r}_{k}^{(2)}\in N_{C_2}(\bm{\pi}_{k+1})$ and $\kappa_2 > 0$.  The existence of $\bpi_{k+1}$ can be shown from \cite[Rem. 3.3]{Chouzenoux_E_2014_j-optim-theory-appl_variable_mfbamsdfcf}. In particular, a minimizer over $C_2$ of the right term in \eqref{eq:maj2} satisfies~\eqref{eq: calc_pi}. The kernel update can be deactivated, if the kernel is known (i.e., non blind case), Algorithm~1 then identifies with~ \cite{Cherni_A_2020_j-ieee-tsp_spoq_lpolqrssrams}.

\subsection{Convergence Result}
We establish the following convergence theorem for~ Algorithm~\ref{algo:TR-BC-VMFB}. Its proof  is provided in the supplementary material.

\begin{theorem}
\label{prop:conv_BD}
Let~$(\bm{s}_k)_{k \in \mathbb{N}}$ and~$(\bm{\pi}_k)_{k \in \mathbb{N}}$ be sequences generated by Alg.~\ref{algo:TR-BC-VMFB}. If $(C_1,C_2)$ are semi-algebraic sets, and $\nabla f$ is Lipschitz on the domain of $\Omega$, then the sequence~$(\bm{s}_{k}, \bm{\pi}_k)_{k \in \mathbb{N}}$ converges to a critical point~$(\widehat{\bm{s}}, \widehat{\bm{\pi}})$ of Problem~\eqref{pb: pb2solve}. 
\end{theorem}
The above result is novel, as it extends \cite[Theo.1]{Cherni_A_2020_j-ieee-tsp_spoq_lpolqrssrams} to the block alternating case using proof ingredients from~\cite{Chouzenoux_E_2016_j-global-optim_block_cvmfba,Hien_L_2020_p-icml_intertial_bpmncnso}. The assumption on $(C_1,C_2)$ ensures that function $\Omega$ satisfies Kurdyka-\L{}ojasiewicz inequality, which is essential for the proof of descent schemes in a non-convex setting \cite{Bolte_J_2014_j-math-programm_proximal_almnnp}.

\section{Numerical results\label{Sec_Results}}

\subsection{Datasets}
Two datasets A and B were considered. The original sparse signal~$\overline{\bs}$ and the observed signal~$\by$ are shown in Fig.~\ref{Fig_Dataset}, both of size $N = 200$.
Signal~$\by$ is obtained from~\eqref{eq:observation_signal} where $\overline{\bpi}$ is a normalized Gaussian kernel with standard deviation 0.15 and size~$L = 21$. The noise $\bm{n}$ is zero-mean white Gaussian with variance $\sigma^2$ either equals~$\SI{0.5}{\percent}$ or~$\SI{1}{\percent}$ of $x_{\max}$ defined as the maximum amplitude of~$\overline{\bm{x}} = \overline{\bpi} \ast \overline{\bs}$. Signal and kernel convolution is implemented with zero padding. 
Trend $\overline{\bt}$ is taken as the low-frequency signal from \cite{Ning_X_2014_j-chemometr-intell-lab-syst_chromatogram_bedusbeads}.

\subsection{Algorithmic settings}
We set $C_1 = [0,100]^N$, and~$C_2$ the simplex unit set, i.e.~$C_2 = \!\{\bpi \!=\! (\pi_{\ell})_{1 \leq \ell \leq L} \in [0,+\infty[^L \; \text{s.t.} \; \sum_{\ell=1}^L \pi_{\ell} = 1  \}$. For such choices, the assumptions of Theorem~\ref{prop:conv_BD} hold, and since metric \eqref{eq:A_qrho} is diagonal, the resolution of~\eqref{eq: calc_s} and~\eqref{eq: calc_pi} is straightforward, by \cite[Prop. 24.11]{Bauschke_H_2011_book_convex_amoths} and \cite[Cor. 9]{Becker_S_2012_p-nips_quasi-newton_psm}. Namely, for every $k \in \mathbb{N}$, and $i \in \{1,\ldots,\mathcal{I}\}$,
\begin{equation*}
\begin{cases}
     \bs_{k,i} \!=\! \text{Proj}_{ C_1}\!\!\left( \bs_k \!-\! \gamma_{s,k}\bm{A}_{1,\rho_{k,i}}(\bm{s}_{k},\bm{\pi}_{k})^{-1}\nabla_1 f(\bs_k,\bpi_k) \right),\\
    \bpi_{k+1} = \text{Proj}_{C_2}\left(
    \bpi_k - \gamma_{\pi,k}\Lambda_{2}(\bm{s}_{k+1})^{-1}\nabla_2 f(\bs_{k+1},\bpi_k)
    \right).
    \end{cases}
\end{equation*}
Hereabove, 
$\text{Proj}_{ C_1}$ 
is the projection over the positive orthant, that has a simple closed form expression, while $\text{Proj}_{C_2}$ is the projection over the simplex unit set, that can be computed using the fast procedure from \cite{Condat_L_2016_j-math-programm_fast_psl1b}. For simplicity, we set constant stepsizes $ \gamma_{s,k} \equiv \gamma_{\pi,k} \equiv 1.9$, thus satisfying the required range assumption. Moreover, we take $\theta = 0.5$ in the TR update, and a maximum of $\mathcal{I} = 50$ of TR trials. We use the same initialization strategy for all methods as in~\cite{Repetti_A_2015_j-ieee-spl_euclid_tsbdsl1l2r}, namely~$\bs_0 \in C_1$ is a constant positive-valued signal and~$\bpi_0 \in C_2$ is a centered Gaussian filter with standard deviation of 1. The stopping criterion parameters are set as~$\varepsilon = \sqrt{N}\times10^{-6}$ and $K_{\max} = 2000$.

\subsection{Numerical results}
\pend jointly performs  blind deconvolution and trend removal, using SPOQ penalty. Let us recall that SOOT penalty from~\cite{Repetti_A_2015_j-ieee-spl_euclid_tsbdsl1l2r} is retrieved by setting~$(p,q) = (1,2)$ in SPOQ. Another setting will be analyzed, namely $(p,q) = (0.75,2)$. Other choices led to similar or poorer restoration results, as also observed in \cite{Cherni_A_2020_j-ieee-tsp_spoq_lpolqrssrams}. In the spirit of an ablation study, we compare PENDANTSS pipeline with the state-of-the-art background estimation method backcor~\cite{Mazet_V_2005_j-chemometr-intell-lab-syst_background_rsdmnqcf} to estimate and remove the  {trend}, followed by the blind deconvolution method~\cite{Repetti_A_2015_j-ieee-spl_euclid_tsbdsl1l2r}, to estimate the signal~$\widehat{\bs}$ and the kernel~$\widehat{\bpi}$. In both cases, we either use  SPOQ $(p,q) = (0.75,2)$, or SPOQ $(p,q) = (1,2)$ (i.e., SOOT) for promoting sparsity in $\widehat{\bs}$.

We use  signal-to-noise ratios to evaluate our estimations, respectively for signal (SNR$_{\bs}$), kernel (SNR$_{\bpi}$) and trend (SNR$_{\bt}$). For instance,~SNR$_{\bs} = 20\log_{10}(\|\overline{\bs}\|_2 / \|\overline{\bs} - \widehat{\bs}\|_2)$. 
Moreover, $\text{TSNR}$ evaluates the~SNR only on the support of the original sparse signal. While their support are not known in general, it reveals how peak-derived quantities (height, width, area), important for downstream quantitative chemical analysis, would be impacted by detrending and deconvolution.

Hyperparameters, e.g.  regularization parameters of backcor~\cite{Mazet_V_2005_j-chemometr-intell-lab-syst_background_rsdmnqcf} and  SPOQ/SOOT parameters~$(\lambda, \beta,\eta)$, are adjusted through grid search to maximize a weighted sum of SNRs for one completely known reference realization, i.e.~$2\text{SNR}_{\bs} + \text{SNR}_{\bpi} + \text{SNR}_{\bt}$, which appeared as a representative metric in our experiments. We set $\alpha=7 \times 10^{-7}$ as recommended in~\cite{Cherni_A_2020_j-ieee-tsp_spoq_lpolqrssrams}. In practice, $(\alpha, \beta,\eta)$ have little influence on performance, while the choice of $\lambda$ is critical. The cutoff frequency of the low-pass filter in~\eqref{eq:t_hat_orig} is chosen as the best performing point over the first ten peak points of the modulus of the signal frequency spectrum. 
To assure the kernel is centered, a spatial shift on the estimated kernel and the sparse signal is applied as a post-processing step as spatially shifted kernels and sparse signals result in the same observed signal.
A rough grid search determines the number of inner loops to maximize the~$\text{SNR}_{\bs}$. 

Table~\ref{tab: results} summarizes the results{ of mean SNR values, and standard deviations after the ``$\pm$" sign},  {calculated} over two hundred noise realizations. 
 Best and second best values are almost always achieved by the proposed \pend approach with $(p,q) = (0.75,2)$ or $(1,2)$. The difference with the baseline methods is also significant for all cases especially in terms of TSNR$_{\bs}$ and SNR$_{\bt}$. One exception lies on SNR$_{\bpi}$ with dataset~B with the noise level of~\SI{1}{\percent} of~$x_{\max}$, where the second best is achieved by the combination  backcor+SPOQ.
 We stress out that in such problems, correct estimations of sparse signal and baseline are usually more important than kernel estimation.  

 Regarding parameters $(p,q)$, the performance of \pend  is dependent on the datasets and the noise level.\deleted{
In terms of sparse signal recovery SNR$_{\bs}$ and TSNR$_{\bs}$, \pend with $(p,q) = (0.75,2)$ achieves higher performance than \pend with $(p,q) = (1,2)$ for dataset A. However, its outcomes are notably lower for dataset B, a less sparse signal, while remaining the second best method.
  For dataset A, both \pend methods have similar baseline estimation accuracy, while for dataset B, \pend $(p,q) = (0.75,2)$ performs better with lower noise level and \pend $(p,q) = (1,2)$ better with greater noise level, with a difference of SNR of about~$2$ dB. As for the estimation of SNR$_{\bpi}$, \pend with $(p,q) = (1,2)$ performs the best for all four cases with little difference for dataset A but a larger difference for more challenging cases with dataset B and higher noise levels.}
Considering various SPOQ parameters is indeed beneficial. According to the presented simulation results, \pend with $(p,q) = (0.75,2)$ is better for datasets with sparser, well-separable peaks (dataset A) whereas \pend with $(p,q) = (1,2)$ is preferable for more challenging datasets (dataset B). Graphical details on the quality of estimated peaks are provided as supplementary material. 
Computational cost for  \pend is slightly higher than for the sequential method with backcor: in the order of 4 s. vs 1 s. for dataset A and 20 s. vs 10 s. for dataset B on a standard laptop.

 \begin{figure}[t]
\centering
 \subfloat[Dataset A.\label{fig:dataA}]{
  \includegraphics[width=0.75\linewidth]{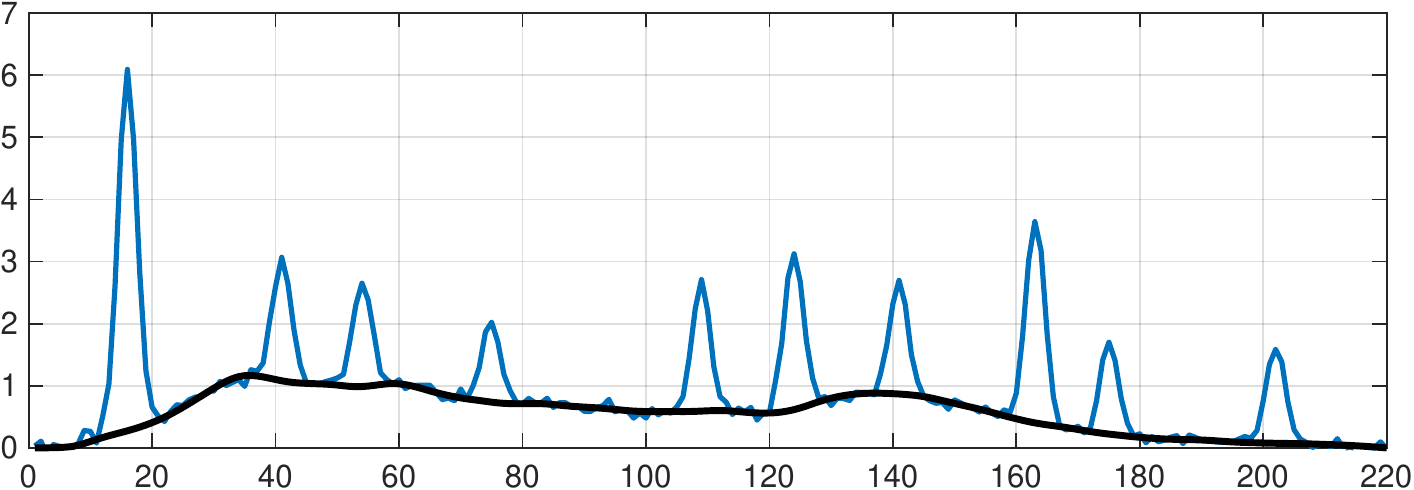}}
    \vspace{-.2cm}
  \hfill 
  \subfloat[Sparse spike signal for dataset A.\label{fig:dataA}]{
  \includegraphics[width=0.75\linewidth]{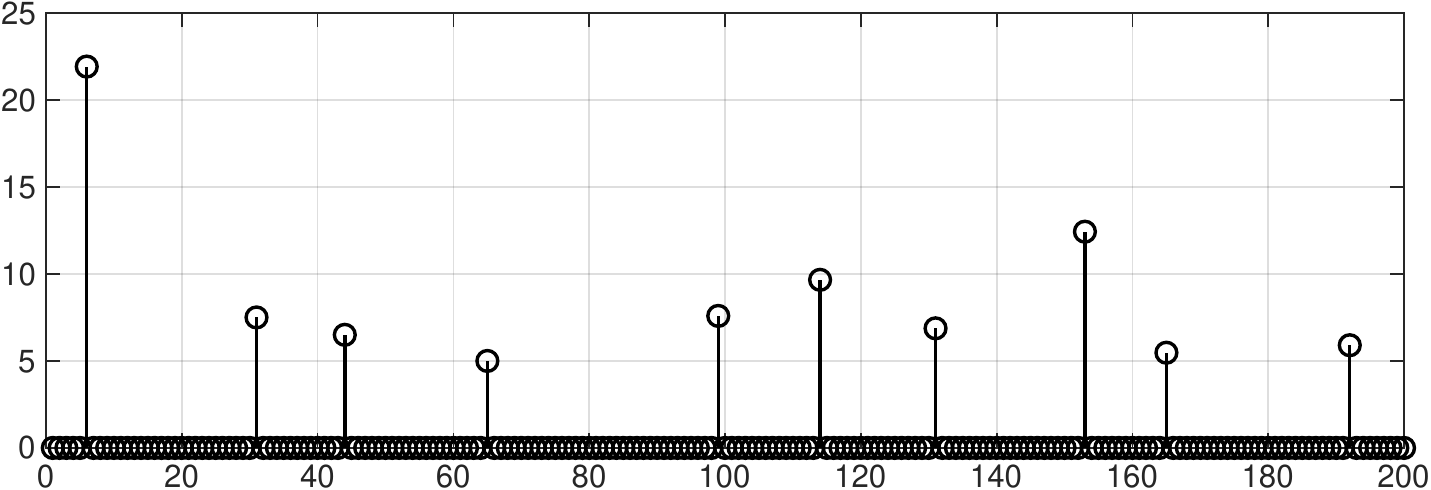}}
  \vspace{-.2cm}
\hfill 
 \subfloat[Dataset B.\label{fig:dataB}]{
 \includegraphics[width = 0.75\linewidth]{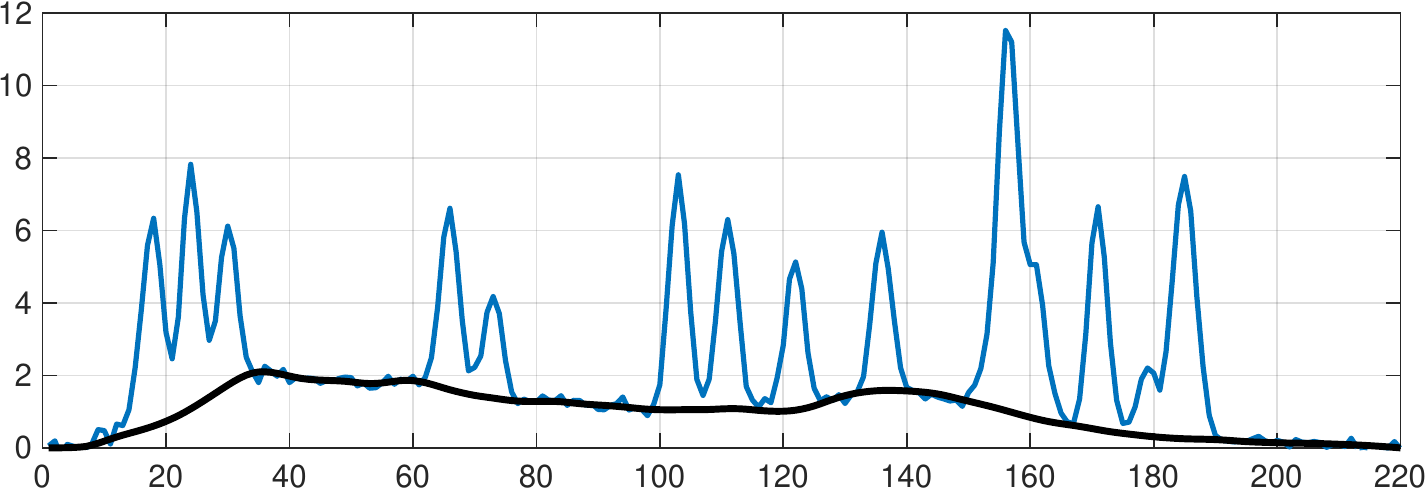}}
  \vspace{-.2cm}
 \hfill 
  \subfloat[Sparse spike signal for dataset B.\label{fig:dataB}]{
 \includegraphics[width = 0.75\linewidth]{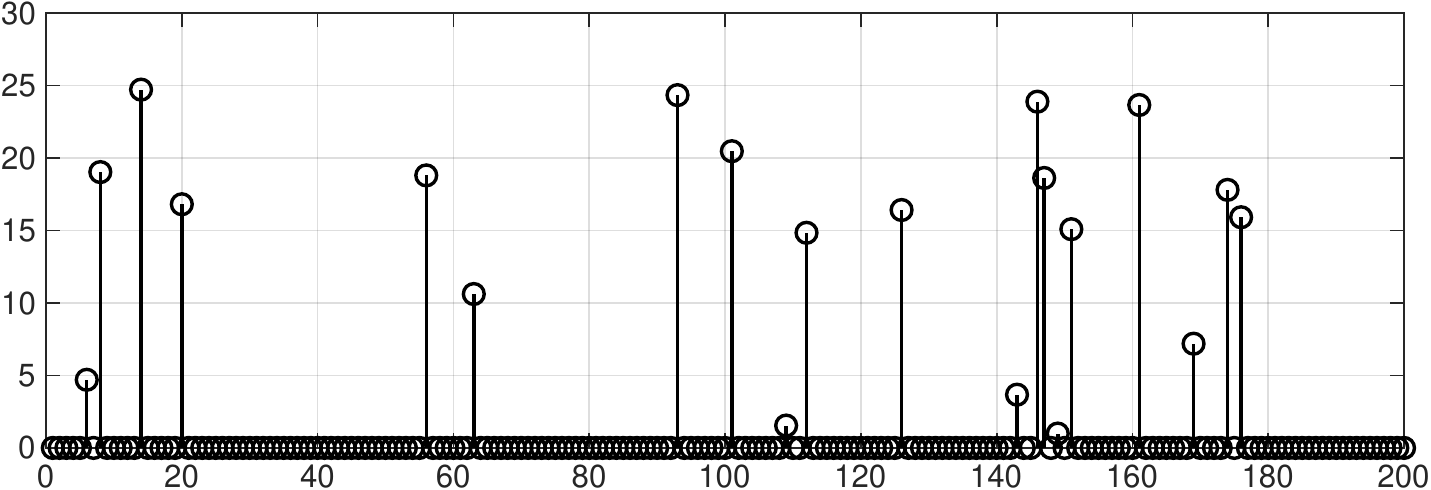}}
 \vspace{-.1cm}
 \caption{Unknown sparse signal~$\overline{\bs}$ (b) and (d); in (a) and (c) observation~$\bm{y}$ (blue) and baseline $\overline{\bt}$ (black) (bottom) for datasets A and B. Signal A has 10 spikes (\SI{5}{\percent}  of sparsity) while signal B has 20 spikes (\SI{10}{\percent} of sparsity).}
\label{Fig_Dataset}
       \vspace{-.2cm}
\end{figure}

\section{Conclusion and perspectives}

We address a complicated joint sparse signal blind deconvolution and additive trend  problem. Our method handles smooth  trend removal by exploiting the low-pass property  and simplifies the problem into a blind deconvolution problem formulation integrating the SPOQ sparse penalty and appropriate constraints. A new block alternating algorithm with trust region acceleration is introduced, and its convergence is established. \pend  outperforms comparable methods on typical sparse analytical signals on  simulation results. 
Further works include its validation on  other sparse spike signals. The appropriate parameters for the sparsity-promoting norm ratio penalty ought to be investigated, for instance with respect to the alleged signal sparsity or peak separability. \pend Matlab code and  hyper-parameter extensive analysis are available at \url{https://github.com/paulzhengfr/PENDANTSS}. The authors thank Vincent Mazet, Bruno Lety, the reviewers and the associate editor.

%
\begin{table}[t]
\tiny
\renewcommand{\arraystretch}{1.3}
{
    \centering
        \caption{Numerical results on datasets A and B. SNR quantities in dB. Best performing method followed by~**, second by~*. \label{tab: results}}
        \vspace{-.2cm}
         \centering
    \begin{tabular}{|c| r  | r | r  | r| r | }
    \cline{3-6}
    \multicolumn{2}{c}{}& \multicolumn{2}{|c|}{Dataset A}   & \multicolumn{2}{|c|}{Dataset B}  \\ \hline 
    \multicolumn{2}{|c|}{Noise level $\sigma$ (\SI{}{\percent} of $x_{\max})$}& \SI{0.5}{\percent}\pS\pS & \SI{1}{\percent}\pS\pS   & \SI{0.5}{\percent}\pS\pS & \SI{1}{\percent}\pS\pS  \\ \hline\hline
 
    
    \parbox[t]{2mm}{\multirow{4}{*}{\rotatebox[origin=c]{90}{SNR$_{\bs}$}}}

    & backcor+SOOT   & \num{29.155080}$\pm$\num{0.734889}\pS\pS &  \num{28.515364}$\pm$\num{1.908669}\pS\pS & 
    \num{14.873117}$\pm$\num{4.033792}\pS\pS & \num{11.522091}$\pm$\num{4.743623}\pS\pS 

  \\\cline{2-6}
     & backcor+SPOQ   &\num{29.202738}$\pm$\num{0.734030}\pS\pS &  \num{29.314277}$\pm$\num{1.316920}\pS\pS &  \num{12.902406}$\pm$\num{3.494131}\pS\pS & \num{11.320521}$\pm$\num{4.412080}\pS\pS 

  \\\cline{2-6}
     & PENDANTS (1, 2)   & \num{32.948162}$\pm$\num{1.489943}*\pS  &  \num{30.868469}$\pm$\num{2.184722}*\pS   &   \num{22.294528}$\pm$\num{8.169975}** & \num{17.461974}$\pm$\num{8.350769}**

  \\\cline{2-6}
      & PENDANTS (0.75, 2)    &       \num{33.239739}$\pm$\num{2.297935}** &   \num{31.010023}$\pm$\num{4.225174}** & \num{15.934151}$\pm$\num{4.505353}*\pS  & \num{12.866945}$\pm$\num{4.570091}*\pS  


  \\
     \hline \hline
     
      \parbox[t]{2mm}{\multirow{4}{*}{\rotatebox[origin=c]{90}{TSNR$_{\bs}$}}}
      & backcor+SOOT  & \num{29.155080}$\pm$\num{0.734889}\pS\pS &  \num{29.344631}$\pm$\num{1.324279}\pS\pS &  \num{16.602780}$\pm$\num{3.483538}\pS\pS & \num{13.376211}$\pm$\num{4.291049}\pS\pS 
   \\\cline{2-6}
     & backcor+SPOQ   & \num{29.202738}$\pm$\num{0.734030}\pS\pS &  \num{29.314277}$\pm$\num{1.316920}\pS\pS &   \num{15.119774}$\pm$\num{3.036778}\pS\pS & \num{13.703381}$\pm$\num{3.657093}\pS\pS  
  \\\cline{2-6}
      
     & PENDANTS (1, 2)   &  \num{34.099007}$\pm$\num{1.384863}*\pS &  \num{32.157623}$\pm$\num{2.102774}*\pS & \num{24.858240}$\pm$\num{7.975341}** & \num{19.161141}$\pm$\num{7.683192}**

  \\\cline{2-6}
      & PENDANTS (0.75, 2)    & \num{35.431738}$\pm$\num{1.652949}** &  \num{32.623916}$\pm$\num{3.787819}** &  \num{17.689996}$\pm$\num{3.953978}*\pS & \num{14.531254}$\pm$\num{4.052132}*\pS

  \\
     \hline \hline

      \parbox[t]{2mm}{\multirow{4}{*}{\rotatebox[origin=c]{90}{SNR$_{\bt}$}}}
     & backcor+SOOT   & \num{20.467025}$\pm$\num{0.193710}\pS\pS &  \num{20.342458}$\pm$\num{0.351461}\pS\pS  &  \num{15.490505}$\pm$\num{0.505539}\pS\pS & \num{14.787888}$\pm$\num{0.814647}\pS\pS 
\\\cline{2-6}
     & backcor+SPOQ   &  \num{20.467025}$\pm$\num{0.193710}\pS\pS &  \num{20.342458}$\pm$\num{0.351461}\pS\pS &  \num{15.490505}$\pm$\num{0.505539}\pS\pS & \num{14.787888}$\pm$\num{0.814647}\pS\pS  
\\\cline{2-6}
     
      & PENDANTS (1, 2)   & \num{26.934843}$\pm$\num{0.493756}** &  \num{26.047508}$\pm$\num{0.764067}** &  \num{22.034876}$\pm$\num{0.425058}*\pS & \num{21.647655}$\pm$\num{0.966379}**

 \\\cline{2-6}
      &PENDANTS (0.75, 2)     & \num{26.900174}$\pm$\num{0.640798}** &  \num{26.024782}$\pm$\num{1.025717}**  &  \num{24.562554}$\pm$\num{0.610014}** & \num{19.561894}$\pm$\num{3.947416}*\pS

 \\
     \hline \hline
     
     \parbox[t]{2mm}{\multirow{4}{*}{\rotatebox[origin=c]{90}{SNR$_{\bpi}$}}}
      & backcor+SOOT   & \num{36.309548}$\pm$\num{1.349624}\pS\pS &  \num{33.883735}$\pm$\num{1.650252}\pS\pS &  \num{30.279319}$\pm$\num{1.265472}\pS\pS & \num{28.498666}$\pm$\num{1.776255}\pS\pS  
 \\\cline{2-6}
      & backcor+SPOQ   & \num{36.250188}$\pm$\num{1.310875}\pS\pS &  \num{33.950011}$\pm$\num{1.671105}\pS\pS &  \num{33.109379}$\pm$\num{1.924132}\pS\pS &  \num{31.238245}$\pm$\num{2.141124}*\pS  
 \\\cline{2-6}
     
      & PENDANTS (1, 2)  & \num{41.302457}$\pm$\num{2.004572}** &   \num{34.422905}$\pm$\num{2.428033}** &   \num{38.329523}$\pm$\num{1.850341}** & \num{33.577125}$\pm$\num{2.244100}**

 \\\cline{2-6}
      & PENDANTS (0.75, 2)     &
      \num{41.258712}$\pm$\num{1.955278}** &  \num{34.227444}$\pm$\num{2.451326}*\pS  &  \num{35.663513}$\pm$\num{1.508100}*\pS & \num{25.378523}$\pm$\num{5.537551}\pS\pS

 \\
     \hline
    \end{tabular}
    }
    \vspace{-.4cm}
\end{table}

\newpage
\bibliographystyle{IEEEtran}
\bibliography{abbr,Reference_Base_DeNovo,Reference_New}

\end{document}


{
\small
\title{PENDANTSS: Supplementary Material}

\maketitle
}
\section{Proof of Theorem 1 for Algorithm 1}

We first provide a useful    majorant metric matrix property.

\begin{lemma}
There exists  $(\underline{\lambda}, \overline{\lambda}) \in ]0,+\infty[^2$ such that for every $k\in \N$, and for every  $i\in\{1,\ldots, \mathcal{I} \}$, 
\begin{AEquation}
\begin{cases}
    \underline{\lambda}\mathrm{\mathbf{Id}}_N \preceq \bm{A}_{1,\rho_{k,i}}(\bm{s}_{k},\bm{\pi}_{k}) \preceq \overline{\lambda}\mathrm{\mathbf{Id}}_N,\\
   \underline{\lambda} \leq   \Lambda_2(\bm{s}_{k}) \leq \overline{\lambda}.
    \end{cases}
\end{AEquation}
\label{lem:1}
\end{lemma}
\begin{proof}
Direct consequence of [14, Prop. 2] and [13, Prop. 1]. 
\end{proof}
We then show that Algorithm~1 satisfies two essential descent properties, that are key for the convergence analysis. 

\begin{lemma}
\label{prop:descent}
There exists $(\mu_1,\mu_2) \in ]0,+\infty[^2$ such that, for every $k\in \mathbb{N}$, the following descent properties hold:
\begin{AEquation}
\Omega(\bs_{k+1}, \bpi_{k}) \leq\Omega(\bs_{k}, \bpi_k) - \frac{\mu_1}{2}||\bs_{k+1}- \bs_{k}||^2,
\label{eq:descent_s}
\end{AEquation}
\begin{AEquation}
	\Omega(\bs_{k\!+\!1}, \bpi_{k\!+\!1})\! \leq\! \Omega(\bs_{k+1}, \bpi_{k}) - \frac{\mu_2}{2}|| \bpi_{k\!+\!1} -  \bpi_{k}||^2.
\label{eq:descent_pi}
\end{AEquation}
\end{lemma}

\begin{proof}
Let $k \in \mathbb{N}$. We remind that the objective function $\Omega$ is defined in~(12), with $g$ specified in (11). By construction, $\bs_{k+1}\in \bar{\mathcal{B}}_{q,\rho} \cap C_1$ for some $i\in\{1,\ldots, \mathcal{I} \}$. Summing the majoration~(14) and the first inequality in~(17) yields:
\begin{equation*}
\Omega(\bs_{k+1},\bpi_k)\leq f(\bs_k, \bpi_k) - (\gamma_{s,k}^{-1} - \tfrac{1}{2})\| \bs_k - \bs_{k+1} \|^2_{\bm{A}_{1,\rho}(\bs_k,\bpi_k)}.
\end{equation*}
We notice that~$f(\bs_k, \bpi_k)= \Omega(\bs_k, \bpi_k)$ since $\bs_k\in C_1$ and~$\bpi_k\in C_2$.  Using Lemma~1 and the range assumption on $\gamma_{s,k}$ allows to show \eqref{eq:descent_s} for~$\mu_1 = {\underline{\lambda}\overline{\gamma}}/({2-\overline{\gamma}})$.
Again by construction, $\bpi_{k+1}\in C_2$. Summing~(18) and~(19) leads to:
\begin{multline*}
\Omega(\bs_{k+1}, \bpi_{k+1}) \leq f(\bs_{k+1},\bpi_k) -\\ (\gamma_{\pi,k}^{-1} - \tfrac{1}{2}) \Lambda_2(\bs_{k+1}) \|\bm{\pi}_{k +1}  - \bm{\pi}_{k }\|^2.
\end{multline*}
Here again, we use~$f(\bs_{k+1}, \bpi_k)= \Omega(\bs_{k+1}, \bpi_k)$ as $\bs_{k+1}\in C_1$ and~$\bpi_k\in C_2$. The descent property \eqref{eq:descent_pi} is obtained by using Lemma~\ref{lem:1}, the range constraint on $\gamma_{\pi,k}$, and setting~$\mu_2 = {\underline{\lambda}\bar{\gamma}}{(2-\bar{\gamma})}$.
\end{proof}

The rest of the proof of Theorem 1 is obtained by following the same lines than the one of~[16, Theorem 3.1], \textcolor{black}{leveraging the Lipschitz smoothness of $f$ on the domain $C_1 \times C_2$ of $\Omega$, and the Kurdyka-\L{}ojasiewicz inequality satisfied by $\Omega$.}

\section{Additional results}
Figures  \ref{fig:estimation} and  \ref{fig:estimation2} provide additional insights into PENDANTSS restoration. Dataset A in Figure  \ref{fig:estimation}-(a) presents sparse and well-isolated peaks. Accurate peak restoration is secured. Peak shapes are well recovered (left-hand zoom), and the estimated trend matches well the actual baseline. As a consequence, peak features that are computed  with respect to the trend (height, area) are likely to be well-estimated with PENDANTSS.
The less sparse Dataset B in Figure  \ref{fig:estimation}-(b) shows that the separation and the height of close peaks are accurately matched. Some overshoot in trend estimation can  be noticed. It is however not likely to drastically affect relative peak height or area computations.

Retrieved spikes are exposed in Figure \ref{fig:estimation2}. For Dataset A, well-separated spikes are accurately recovered using PENDANTSS. Estimated amplitudes and locations are almost indistinguishable from the original ones. This is exemplified for the less sparse Dataset B in Figure \ref{fig:estimation2}-(b).  Isolated peaks are well-estimated. However, some spikes (for instance around  index $175$) for  Dataset B in Figure \ref{fig:estimation2}-(b) remain unelucidated. Three contiguous spikes are estimated, instead of two. Such an ambiguous solution is typical to source separation problems.

\begin{figure}[htb]
    \centering
    \subfloat[Dataset A - reconstruction and trend.\label{fig:dataA_trend}]{
  \includegraphics[width=0.67\linewidth, trim= 0 0 0 20]{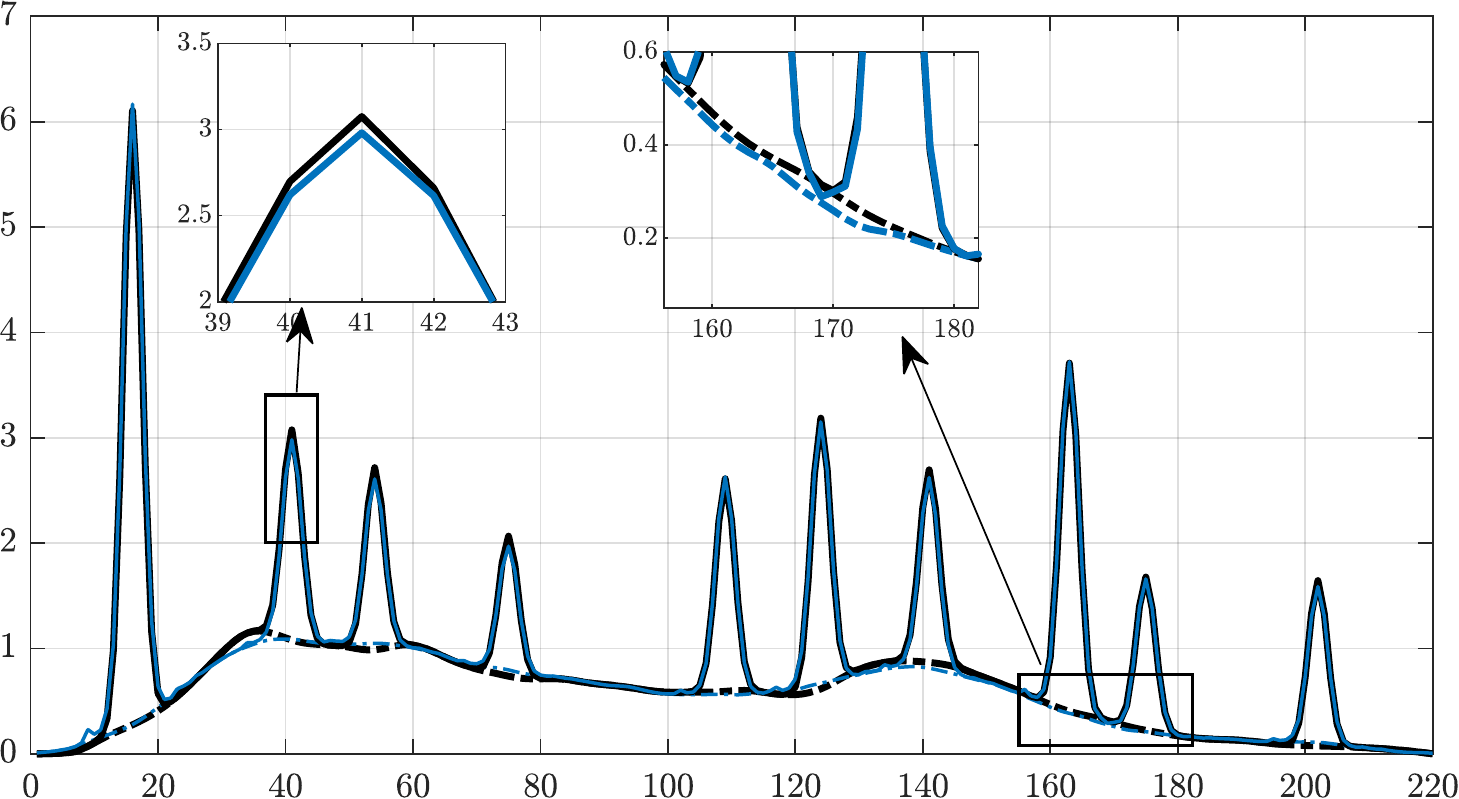}}
  \vspace{-.2cm}
\hfill 
    \subfloat[Dataset B - reconstruction and trend.\label{fig:dataB_trend}]{
  \includegraphics[width=0.67\linewidth]{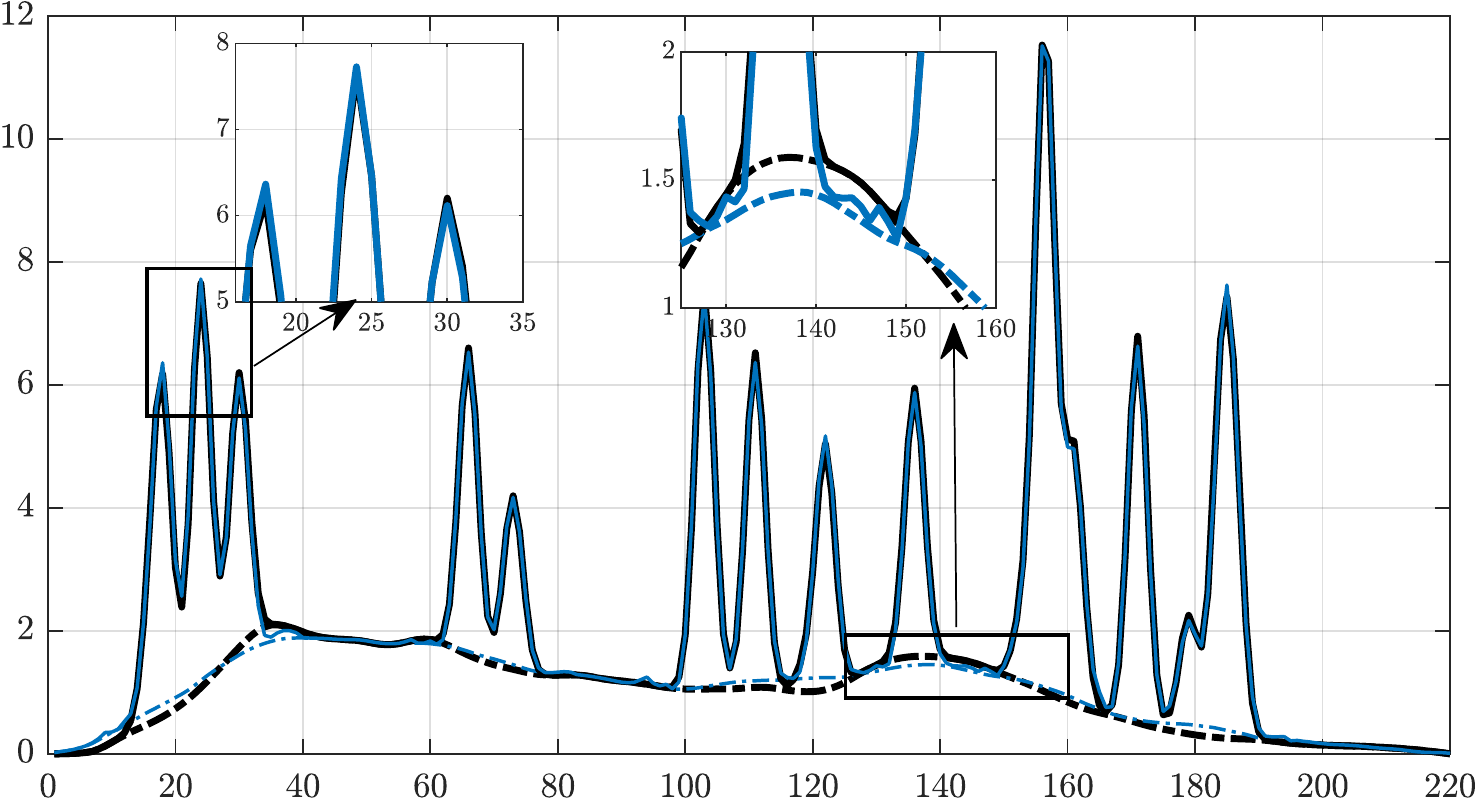}}
  \vspace{-.2cm}
\hfill 

    \caption{Ground truth (thick black line) and proposed estimation results (thin blue line) for the baseline~$\bt$ (dashed dot) and the signal~$\bs \ast \bpi$ (continuous).}
    \label{fig:estimation}
\end{figure}

\begin{figure}[htb]
    \centering
\subfloat[Dataset A - sparse spike  signal.\label{fig:dataA_sparse}]{
  \includegraphics[width=0.67\linewidth, trim= 0 0 0 30 ]{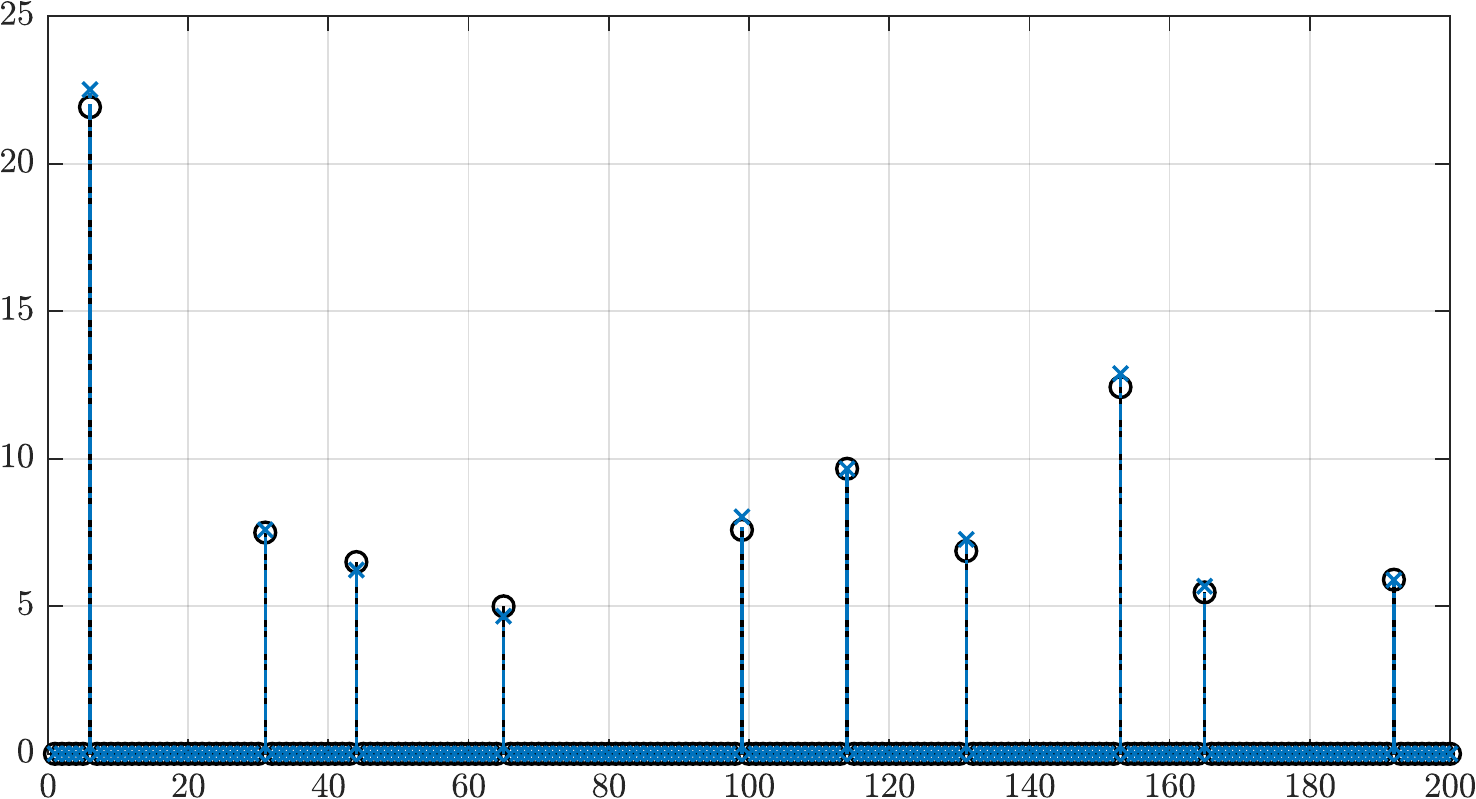}}
  \vspace{-.2cm}
\hfill 
\subfloat[Dataset B - sparse spike signal.\label{fig:dataB_sparse}]{
  \includegraphics[width=0.67\linewidth]{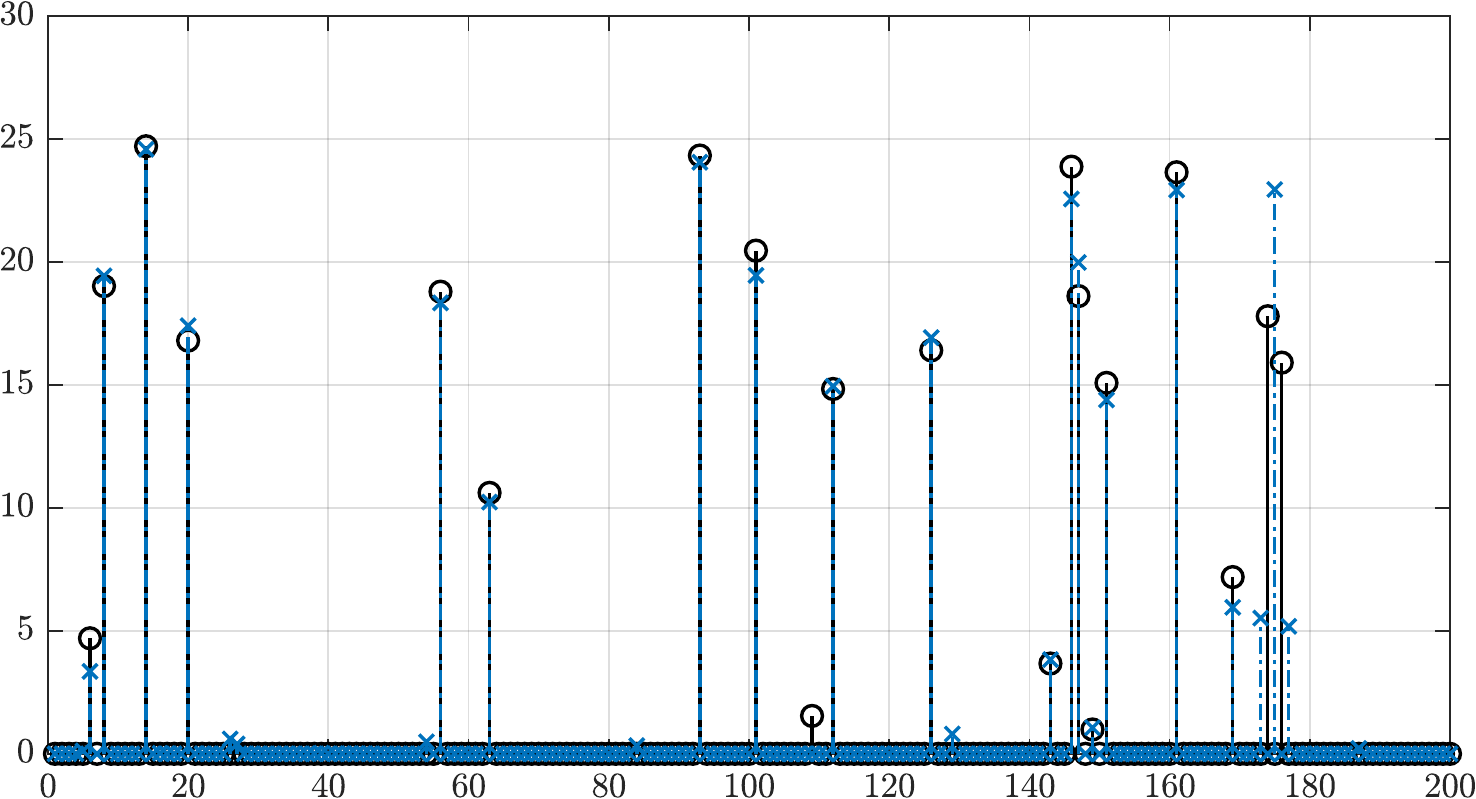}}
  \vspace{-.2cm}
\hfill 
    \caption{Ground truth (black line with circle marker) and proposed estimation results (blue line with cross marker) for sparse spike signal $\bs$.}
    \label{fig:estimation2}
    \vspace{-.5cm}
\end{figure}